\documentclass[12pt]{amsart}
\usepackage[utf8]{inputenc}
\usepackage[english]{babel}
\usepackage{multicol}
\usepackage[english]{babel}
\usepackage{amsmath, amssymb, amsthm, amscd,color,comment}
\usepackage{cancel}
\usepackage{cite}
\usepackage{alltt}
\usepackage[dvipsnames]{xcolor}
\usepackage{array}
\usepackage[small,bf,labelsep=period]{caption}
\usepackage{mathtools}
\usepackage{hyperref}
\usepackage{enumerate}

\usepackage{pgfplots}
\pgfplotsset{compat=1.15}
\usepackage{mathrsfs}
\usetikzlibrary{arrows}
\usepackage{hyperref}

\usepackage{enumitem}

\usepackage{algorithm}
\usepackage{algpseudocode}

\usepackage{enumitem}

\usepackage{enumerate}
\newtheorem{theorem}{Theorem}[section]

\newtheorem{example}[theorem]{Example}
\newtheorem{remark}[theorem]{Remark}
\newtheorem{lemma}[theorem]{Lemma}

\newtheorem{proposition}[theorem]{Proposition}
\newtheorem{assumption}[theorem]{Assumption}

\begin{document}
\title{Permutation Decoding of AG Codes from Curves Defined by Separated Polynomials}

\thanks{{\bf Keywords}: Permutation Decoding, AG Codes, Curves from Separated Polynomials.}

\thanks{{\bf Mathematics Subject Classification (2020)}: 94B05, 11T71, 14G50}

\author{Alonso S. Castellanos, Guilherme Tizziotti, and Wilson Olaya-Le\'on}

\address{Instituto de Matem\'atica e Estat\'istica, Universidade Federal de Uberl\^andia, Campus Santa M\^onica, CEP 38400-902, Uberl\^andia, Brazil}
\email{alonso.castellanos@ufu.br}

\address{Instituto de Matem\'atica e Estat\'istica, Universidade Federal de Uberl\^andia, Campus Santa M\^onica, CEP 38400-902, Uberl\^andia, Brazil}
\email{guilhermect@ufu.br}

\address{Escuela de Matem\'aticas, Universidad Industrial de Santander, Carrera 27 Calle 9, AA 678, Bucaramanga, Santander, Colombia}
\email{wolaya@uis.edu.co}

\begin{abstract}
   In this work, we investigate permutation decoding for algebraic geometry (AG) codes  arising from algebraic curves defined by separated polynomials. Using automorphisms of the underlying curves, we construct permutation automorphisms of the associated AG codes and exploit the resulting orbit structure to determine information and check positions. We introduce a class of curves, called SAP curves (Separated Additive Polynomial curves), and investigate one-point AG codes defined on them. For these codes, we obtain permutation decoding sets that correct burst errors supported on coordinates associated with rational points sharing a common coordinate. We further identify a subclass of special SAP curves, including Hermitian curves, generalized Hermitian curves, and certain maximal curves, for which additional automorphisms yield more powerful decoding sets.
\end{abstract}

\maketitle

\section*{Introduction}

Coding theory plays a fundamental role in modern digital communication and data storage systems. Its main purpose is to provide reliable transmission of information through noisy channels by introducing redundancy into the transmitted data \cite{Mac1977}. Among the various classes of error-correcting codes, linear codes over finite fields constitute one of the most important and widely studied families due to their rich algebraic structure and efficient decoding algorithms \cite{HP2003}.

Let $\mathbb{F}_q$ denote the finite field with $q$ elements, where $q$ is a prime power. An $[n,k,d]_q$ linear code $C$ is a $k$-dimensional vector subspace of $\mathbb{F}_q^n$ with minimum Hamming distance $d$. The parameters $n$, $k$, and $d$ correspond respectively to the code length, dimension, and minimum distance. Linear codes over finite fields have numerous applications in communication systems, satellite transmission, distributed storage, cryptography, and deep-space communication \cite{HP2003}.

One of the central problems in coding theory is the decoding process, namely, determining the original transmitted codeword from a received vector possibly affected by errors. Although maximum likelihood decoding is optimal, it is computationally difficult in general. Consequently, several practical decoding techniques have been developed for specific families of codes \cite{R1992}.

Permutation decoding is an efficient algebraic decoding method introduced by MacWilliams \cite{Mac1964} for binary linear codes. The method exploits permutations from the automorphism group of the code to move error positions outside a fixed information set. More precisely, if $C$ is a linear code with automorphism group $\mathrm{Aut}(C)$, one searches for a set of permutations, called a PD-set (permutation decoding set), such that every correctable error pattern can be displaced from the information positions by at least one permutation in the set \cite{HP2003}.

The main advantage of permutation decoding lies in its ability to simplify the decoding procedure by combining syndrome decoding with the symmetry properties of the code. This approach has been successfully applied to several important families of codes, including cyclic codes, Reed--Muller codes, Hadamard codes, and algebraic geometry codes \cite{LLMS2026,G2006,Gordon1982}. Furthermore, the construction of small PD-sets remains an active research topic due to its direct impact on decoding efficiency and computational complexity \cite{G2006}.

Hence, codes possessing large automorphism groups are particularly well suited for this decoding technique. This observation naturally motivates the study of AG codes arising from curves over finite fields defined by separated polynomials, since many such curves are known to admit rich and highly structured automorphism groups. In \cite{LLMS2026}, Lichtenwalner, López, Matthews and Seneviratne investigated permutation decoding for AG codes arising from Hermitian and Norm-trace curves.

In this paper, we consider permutation decoding of AG codes on curves defined by separated polynomials, generalizing the results presented in \cite{LLMS2026}. The paper is organized as follows. Section 1 presents the concepts, notation, and results on permutation decoding and AG codes that will be used throughout the paper. Section 2 introduces certain curves defined by separated polynomials, which are employed in the subsequent sections; it also includes a particular class of these curves whose automorphisms exhibit certain useful properties. Section 3 describes the information position set for certain one-point AG codes. Finally, Section 4 demonstrates how permutation decoding can correct burst errors in one-point AG codes on the curves introduced in the previous section.

\section{Preliminaries}

In this section, we present the basic concepts, notation and results that will be used throughout this paper. Let $\mathbb{F}_q$ denote the finite field with $q$ elements, where $q$ is a power of a prime $p$, and $\mathbb{F}_{q}^{*} = \mathbb{F}_q \setminus \{0\}$. The set $\{1, \dots, n\}$ is denoted by $[n]$, and for any set $A$, its cardinality is denoted by $\#A$.

\subsection{Permutation decoding}

Let $C \subseteq \mathbb{F}_q^n$ be a $[n,k,d]_q$ linear code and let $\sigma \in S_n$ be a permutation. The action of $\sigma$ on a codeword $c \in C$ is given by
\[
\sigma(c) = \big(c_{\sigma^{-1}(1)}, \ldots, c_{\sigma^{-1}(n)}\big).
\]

The permutation automorphism group of $C$ is defined as
\[
\mathrm{Aut}(C) := \{\sigma \in S_n \mid \sigma(C) = C\}.
\]
Throughout this work, the term automorphism will always refer to a permutation automorphism, since monomial transformations are not considered.

The \textit{information positions} are the coordinates of a codeword that contain the symbols corresponding to the original message. The set of indices associated with these coordinates is called an \textit{information set}, while the indices in its complement are referred to as \textit{check positions}. Thus, for a linear $[n,k,d]_q$ code $C$, exactly $k$ coordinates of a codeword $c \in C \subseteq \mathbb{F}_q^n$ determine the original message $c' \in \mathbb{F}_q^k$. Let $G$ be a generator matrix of $C$. For any subset $I \subseteq [n]$, let $G_I$ denote the submatrix of $G$ consisting of the columns indexed by $I$. A subset $I \subseteq [n]$ with $|I|=k$ is called an \textit{information set} of $C$ if $G_I$ is invertible.
So, if the generator matrix $G$ of $C$ is in systematic form,
$
G = [I_k \mid A],
$
then $I=[k]$ is an information set of $C$.

Since permutation decoding depends on the chosen representation of the code, we fix this systematic form for convenience. The arguments below remain valid for any generator matrix in systematic form $G = [I_k \mid A]$ and associated parity-check matrix $H = [-A^T \mid I_{n-k}]$, where $A^T$ is the transpose of the matrix $A$.

A subset $S \subseteq \mathrm{Aut}(C)$ is called an \textit{$r$-PD set} if, for every subset of coordinate positions $\{i_1,\ldots,i_r\} \subseteq [n]$, there exists a permutation $\sigma \in S$ such that
\[
\sigma(\{i_1,\ldots,i_r\}) \subseteq \{k+1,\ldots,n\}.
\]
If $r = \left\lfloor \frac{d-1}{2} \right\rfloor $, then $S$ is simply called a \textit{PD set}. When $r < \left\lfloor\frac{d-1}{2}\right\rfloor$, we refer to $S$ as a \textit{partial $r$-PD set}.

In what follows, we also use a restricted notion of partial PD set, tailored to burst errors supported on a specific subset of positions. Let $\mathcal{Q} \subseteq \mathrm{Supp}(D)$ be a fixed subset of coordinate positions of a one-point AG code $C(D,\lambda P^*)$, with $\#\mathcal{Q}=\nu'$. We say that $S \subseteq \mathrm{Aut}(C)$ is a \textit{partial $r$-PD set for burst errors impacting positions indexed by $\mathcal{Q}$} if, for every subset $E \subseteq \mathcal{Q}$ with $\#E \leq r$, there exists $\sigma \in S$ such that $\sigma(E) \subseteq \{k+1,\ldots,n\}$. When $r < \left\lfloor\frac{d-1}{2}\right\rfloor$, this is strictly weaker than a full $r$-PD set, since it only displaces error positions constrained to lie in $\mathcal{Q}$, not arbitrary positions in $[n]$. In Theorems~\ref{qbgen}--\ref{thm3gen} below, we construct such sets for $\mathcal{Q}=\mathcal{Q}_b$ (points with fixed second coordinate) and $\mathcal{Q}=\mathcal{P}_a$ (points with fixed first coordinate).

\begin{lemma}\label{lem:PD}\cite[Theorem  8.1]{H1998}
Let $C$ be an $[n,k,d]_q$ linear code with information set $I=[k]$ and parity check matrix $H = [-A^T \mid I_{n-k}]$. For any received word $y \in \mathbb{F}_q^n$, the first $k$ coordinates of $y$, denoted by $y_{[k]}$ have no errors if and only if
\[
\mathrm{wt}(Hy^T) \le \left\lfloor \frac{d-1}{2}\right\rfloor,\] where $\mathrm{wt}(\cdot)$ denotes the Hamming weight, i.e. the number of nonzero coordinates.

\end{lemma}

\begin{remark}\label{rem:infset}
If $I\subseteq[n]$ is an arbitrary information set of $C$, one reorders the coordinates so that the positions of $I$ come first; then Lemma~\ref{lem:PD} and the decoding procedure below hold with $[k]$ replaced by $I$ and $\{k+1,\dots,n\}$ replaced by $[n]\setminus I$ (the check positions), and $c$ is recovered from $\psi(y)_I$ through the generator matrix that is systematic on $I$. In particular, $I=[k]$ is \emph{not} required to be an information set of the code in its original ordering.
\end{remark}

Using the previous lemma we get a permutation decoding procedure as follows. Given a PD set $S = \{\sigma_1, \ldots, \sigma_s\}  \subseteq \mathrm{Aut}(C)$ and a received word $y \in \mathbb{F}_q^n$, find $i \in [s]$ such that $\mathrm{wt}\big(H \cdot \sigma_i(y)^T\big) \leq \lfloor \frac{d-1}{2} \rfloor$. Then, decode $y$ as $\sigma_i^{-1}\big((\sigma_i(y))_{[k]}  \cdot G\big)$.

In practical terms, the first step is carried out by testing the above condition for each permutation $\sigma_i$ until one satisfying the condition is found. Once such a permutation has been identified, the first $k$ coordinates of $\sigma_i(y)$ are free of errors. Therefore, any received word $y \in \mathbb{F}_q^n$ with at most $\left\lfloor \frac{d-1}{2}\right\rfloor$ errors can be decoded by computing $\sigma_i^{-1}\big((\sigma_i(y))_{[k]} \cdot G\big).$


\subsection{Algebraic Geometry codes (AG codes)} \label{AG codes section}

Let $\mathcal{X}$ be a projective, non‑singular, geometrically irreducible algebraic curve of genus $g > 0$ over $\mathbb{F}_q$; throughout the paper, we refer to such an object simply as a \textit{curve}.  We denote the function field (field of rational functions) of $\mathcal{X}$ over $\mathbb{F}_q$ by $\mathbb{F}_q(\mathcal{X})$. The set of $\mathbb{F}_q$-rational points of $\mathcal{X}$ is denoted by $\mathcal{X}(\mathbb{F}_q)$. For a given $a, b \in \mathbb{F}_q$, we denote by $\mathcal{P}_a$ the set of all $\mathbb{F}_q$-rational points whose first coordinate equals $a$, and by $\mathcal{Q}_b$ the set of all $\mathbb{F}_q$-rational points whose second coordinate equals $b$; that is,  
\begin{equation} \label{P_Q}
\mathcal{P}_a := \{P = (x,y) \in \mathcal{X}(\mathbb{F}_q) \mid x = a\},  \mbox{ and }
\mathcal{Q}_b := \{P = (x,y) \in \mathcal{X}(\mathbb{F}_q) \mid y = b\}.
\end{equation}
Let $P_1, \ldots, P_n, Q_1, \ldots, Q_\ell$ be $n+\ell$ distinct $\mathbb{F}_q$-rational points on $\mathcal{X}$ and let $m_1, \ldots, m_\ell$ be integers. Consider the divisors $D = P_1 + \cdots + P_n$ and $G = m_1Q_1 + \cdots + m_\ell Q_\ell$. The \textit{Algebraic Geometry code (AG code)} $C(D,G)$ arising from the curve $\mathcal{X}$ is defined as
\[
C(D,G) := \{ (f(P_1), \ldots, f(P_n)) \in \mathbb{F}_q^n \;:\; f \in \mathcal{L}(G) \}, 
\]
where $\mathcal{L}(G)=\{f\in \mathbb{F}_q(\mathcal{X}): (f) \geq -G\}\cup\{0\}$ denotes the Riemann-Roch space associated to divisor $G$.
Recall that $C(D,G)$ is an $[n,\ell(G)-\ell(G-D), \geq n-\mathrm{deg}(G)]$ code where $\ell(G)=\mathrm{dim}_{\mathbb{F}_q}(\mathcal{L}(G))$.
If $G = m_1Q_1$ the AG code $C(D,m_1Q_1)$ is called \emph{one-point AG code}. For more details about AG codes, see e.g.\ \cite{HLP1998, S1993}.

Let $\mathrm{Supp}(D) = \{P_1,\ldots,P_n\}$ be the support of the divisor $D$. Since $\# \mathrm{Supp}(D)=n$, we have that the permutation group $\mathcal{P}(\mathrm{Supp}(D))$ on $\mathrm{Supp}(D)$ is isomorphic to the symmetric group $S_n$, and each $\sigma \in \mathcal{P}(\mathrm{Supp}(D))$ induces an $\mathbb{F}_q$-linear mapping $\tilde{\sigma}$ of the code $C(D,G)$ to $\mathbb{F}_q^n$ by setting
\[
\tilde{\sigma}(f(P_1),\ldots,f(P_n)) := (f(\sigma(P_1)),\ldots,f(\sigma(P_n))).
\]

The mapping $\tilde{\sigma}$ is an automorphism of the code $C(D,G)$ if $\tilde{\sigma}(C(D,G)) = C(D,G)$.

\medskip

In \cite{G1982}, Goppa already observed that the underlying algebraic curve induces automorphisms of the associated AG codes as follows.

\medskip

Let $\mathrm{Aut}(\mathcal{X}) = \{\sigma : \mathcal{X} \to \mathcal{X} : \sigma \text{ is a birational map defined over } \mathbb{F}_q\}$ be the automorphism group of $\mathcal{X}$ over $\mathbb{F}_q$ and, for two divisors $D$ and $G$, consider the subgroup
$$
\mathrm{Aut}_{D,G}(\mathcal{X}):= \{ \sigma \in \mathrm{Aut}(\mathcal{X}) : \sigma(D)=D \text{ and } \sigma(G)=G \}.
$$

Thus, an automorphism in $\mathrm{Aut}_{D,G}(\mathcal{X})$ gives rise to a permutation automorphism of the code $C(D,G)$ as follows. Given $\sigma \in \mathrm{Aut}_{D,G}(\mathcal{X})$, we define $\bar{\sigma} \in S_n$ by
\[
\bar{\sigma} : [n] \to [n], \quad i \mapsto j
\]
if and only if $\sigma(P_i) = P_j$. If $\deg D > 2g - 2$, the map
\[
\varphi : \mathrm{Aut}_{D,G}(\mathcal{X}) \longrightarrow \mathrm{Aut}(C(D,G)), \quad \sigma \mapsto \bar{\sigma}
\]
is injective \cite[Proposition 8.2.3]{S1993}, and the identification $\sigma \leftrightarrow \bar{\sigma}$ yields
\begin{equation}\label{eq:aut_inclusion}
\mathrm{Aut}_{D,G}(\mathcal{X}) \le \mathrm{Aut}(C(D,G)).
\end{equation}

On the other hand, we get a structure of $\mathbb{F}_q[t]$-module on $C(D,G)$ in the following way. 
The set $\mathrm{Supp}(D) = \{P_1, \ldots, P_n\}$ decomposes under the action of 
$\sigma$ into $r$ distinct orbits, $O_1, \ldots, O_r$. For $i = 1, \ldots, r$, fix a 
point $P_{i,0} \in O_i$ and write $P_{i,j} = \sigma^j(P_{i,0}) \in O_i$, 
$j = 0, 1, \ldots, \# O_i - 1$. 
Associated to a codeword $(f(P_1), \ldots, f(P_n)) \in C(D,G)$, $f \in L(G)$, we consider 
the polynomials
\[
h_i(t) = \sum_{j=0}^{|O_i|-1} f(P_{i,j}) t^j, \qquad i = 1, \ldots, r.
\]
In this way we may represent a codeword as an $r$-tuple 
$(h_1(t), \ldots, h_r(t)) \in \mathbb{F}_q[t]^r$, which can be seen also as an 
element of the $\mathbb{F}_q[t]$-module 
$A = \bigoplus_{i=1}^r \mathbb{F}_q[t] / \langle t^{\# O_i} - 1\rangle$. 
The collection $\overline{C}$ of $r$-tuples obtained from all $f \in L(G)$ is 
closed under sum and multiplication by $t$. Define 
$\widetilde{C} = \pi^{-1}(\overline{C})$, where $\pi$ is the natural projection 
from $\mathbb{F}_q[t]^r$ onto $\bigoplus_{i=1}^r \mathbb{F}_q[t] / \langle t^{\# O_i} - 1\rangle$. 
Thus, $C(D,G)$ can be identified to the submodule $\widetilde{C} \subseteq \mathbb{F}_q[t]^r$.

In \cite{HLS1995}, Heegard, Little and Saints presented a systematic encoding for a class of AG codes using Gr\"{o}bner basis for modules, in which certain monomials in $\mathbb{F}_q[t]^r$ are associated with the information positions of the $r$-tuples $(h_1(t), \ldots, h_r(t))$. Subsequently, in \cite{LSH1997}, the same authors introduced an object called a \textit{root diagram}, which provides a Gr\"{o}bner basis to one-point AG codes. The empty boxes in the root diagram determine the dimension of $C(D,m_1Q_1)$ on the curve $\mathcal{X}$ (see \cite[Proposition 2.3]{LSH1997}) and are closely related to the construction of systematic generator matrices and permutation decoding for
one-point Hermitian codes (see \cite[Proposition 3.2 and Theorem 3.5]{LLMS2026}. In the following sections, by exploring the results presented in \cite{FT2018}, we extend the contributions of \cite{LLMS2026} to one-point AG codes arising from certain curves defined by equations with separated variables.


\section{Curves given by Separated Polynomials with special automorphism}

Curves over finite fields defined by separated polynomials constitute a fundamental class of algebraic curves with rich arithmetic and geometric structure. These curves include important examples such as Hermitian and norm-trace curves, which play a central role in the construction of algebraic geometry codes and their automorphism groups. In particular, such curves often exhibit large symmetry groups of automorphism and well-structured sets of rational points, properties that can be effectively exploited in applications to coding theory, see e.g. \cite{BMZ2020, HKT2008, ST2018}.

Let $\mathcal{X}^*$ be the curve defined over $\mathbb{F}_q$ by the affine equation
\[
f(y) = g(x)
\]
satisfying the following conditions:

\begin{enumerate}
\item $f(t), g(t) \in \mathbb{F}_q[t]$, where $\deg(f) = \tau_1$ and $\deg(g) = \tau_2$, with $\gcd(\tau_1,\tau_2)=1$, and $f(t)$ is an additive polynomial over $\mathbb{F}_q$, that is, $f(t)=\sum_{i=0}^{e}a_i t^{p^i}$ with $a_i\in\mathbb{F}_q$ and $a_0\neq 0$; thus $f(X+Y)=f(X)+f(Y)$ in $\mathbb{F}_q[X,Y]$ and $f$ is separable, i.e.\ its roots are distinct;

\item there exists a point $P^*\in\mathcal{X}^*(\mathbb{F}_q)$ such that
\[
\operatorname{div}_\infty(x) = \tau_1P^*, \quad \operatorname{div}_\infty(y) = \tau_2P^*, \quad \text{and} \quad H(P^*) = \langle \tau_1,\tau_2 \rangle,
\]
where $H(P^*)$ is the Weierstrass semigroup at $P^*$, and $\operatorname{div}_\infty(\cdot)$ denote the pole divisor;
\item there exists a divisor $D$ such that $P^* \notin \mathrm{Supp} (D)$ and $\sigma \in \mathrm{Aut}_{D,\lambda P^*}(\mathcal{X^*})$, for some positive integer $\lambda$, given by
$$
\sigma:\begin{cases}
x \mapsto \alpha x \\
y \mapsto  \alpha^{t} y
\end{cases}
$$
for some positive integer $t$ and some $\alpha \in \mathbb{F}_{q}^{*}$.
\end{enumerate}

Note that, by definition of $\sigma$, we have $\mathrm{ord}(\sigma)=\mathrm{ord}(\alpha)$, hence $\mathrm{ord}(\sigma) \mid (q-1)$. 
Since the curve $\mathcal{X}^*$ is defined by separated-variable polynomials, $f$ is an additive polynomial, and $\mathcal{X}^*$ admits an automorphism that fixes an $\mathbb{F}_q$-rational point such that its Weierstrass semigroup is generated by the degrees of $f$ and $g$, from now on we will refer to a curve with these characteristics as \textit{SAP curve}, Separated Additive Polynomial curve.

Examples of SAP curves include several important curves studied in coding theory: 
\begin{itemize}[leftmargin=*, labelsep=0.5em]
\item The Hermitian curves $\mathcal{H}_q$, defined over $\mathbb{F}_{q^2}$ by the affine equation $y^q + y = x^{q+1}$.
\item The Norm-trace curves defined over $\mathbb{F}_{q^{s}}$ by $y^{q^{s-1}} + \cdots + y^q + y = x^{\frac{q^s-1}{q-1}}$.
\item The subcovers of the Hermitian curve, denoted $\mathcal{H}_{q,m}$, defined over $\mathbb{F}_{q^2}$ by the affine equation $y^q + y = x^m$, where $q$ is a prime power and $m > 2$ is a divisor of $q+1$, see \cite{GV1986}.
\item The generalized Hermitian curves $\mathcal{X}_{q,r}$ defined over $\mathbb{F}_{q^{2r}}$ by $y^q + y = x^{q^r+1}$, where $r$ is an odd integer, see \cite{S1973,G1992}. 
\item The Abdón-Torres curves $\mathcal{Y}_{2}$ is a maximal curve defined over $\mathbb{F}_{q^{2}}$ by the affine equation
$y^{{q}/{2}} + y^{{q}/{2^2}} + \cdots + y^2 + y = x^{q+1}$, where $q\geq 4$ be a power of $2$, see \cite{AT1999}. 
\end{itemize}

\begin{remark}
The norm-trace curves and the subcovers $\mathcal{H}_{q,m}$ (with $m < q+1$) are listed above as SAP curves. However, among the results of Section~\ref{PD}, only Theorem~\ref{qbgen} applies to an arbitrary SAP curve satisfying its hypotheses; Theorems~\ref{pagen} and~\ref{thm3gen} require \emph{special} SAP curves, and the norm-trace curves and the subcovers $\mathcal{H}_{q,m}$ with $m<q+1$ are not treated here as special SAP curves. For the norm-trace curves and the subcovers $\mathcal{H}_{q,m}$, the orbit decomposition and the determination of information and check positions via Proposition~\ref{infgen} follow the same pattern as in the examples treated here. We note that \cite{LLMS2026} already handles the Hermitian and norm-trace cases in detail; our contribution is the extension to the broader SAP and special SAP families, illustrated concretely on the generalized Hermitian and Abdón-Torres curves.
\end{remark}

In the next result we describe a structure of some rational points on a SAP curve.

\begin{lemma}\label{tal}
Let  $B=\{\beta\in\mathbb{F}_q: f(\beta)=0\}
$ and let $\mathcal{P}_a$ be as in (\ref{P_Q}). If $(a,b)\in \mathcal{X}^*(\mathbb{F}_q)$, then
\begin{enumerate}
    \item for all $\beta \in B$,  $(a,b+\beta)\in \mathcal{X}^*(\mathbb{F}_q)$.
    \item $\mathcal{P}_a=\{(a,b+\beta): \beta\in B\}$. Consequently, $\#\mathcal{P}_a= \# B$.
\end{enumerate} 
\end{lemma}
\begin{proof}
    (1) Since $f$ is an additive polynomial over $\mathbb{F}_q$, we have that $f(b+\beta )=f(b)+f(\beta)=f(b)=g(a)$.\ \
    (2) By definition of $\mathcal{P}_a$, we have $\{(a,b+\beta): \beta\in B\} \subseteq \mathcal{P}_a$. Now, let $(a,b')\in \mathcal{P}_a$. Thus, $f(b')=g(a)=f(b)\Rightarrow f(b'-b)=0 \Rightarrow b'-b=\beta$ for some $\beta\in B$, and we have $(a,b') \in \{(a,b+\beta): \beta\in B\}$.
\end{proof}

For permutation decoding, we need to know the automorphism group of the code, although sometimes the complete group is not necessary; some set of code automorphisms is sufficient, for example, those provided by the automorphisms of the curve $\mathcal{X}^*$ that fix point $P^*$ and divisor $D$ (see Equation \ref{eq:aut_inclusion}).  In this sense, we present the following result that will be used in the Theorem \ref{qbgen}, where we present a partial $\nu$-PD set that  correct errors in a certain one-point AG code.

\begin{lemma}
Let  $B=\{\beta\in\mathbb{F}_q: f(\beta)=0\}$. Assume that $\mathcal{X}^*$ is a SAP curve with divisor $D$ and $P^* \notin \mathrm{Supp} (D)$ as in condition (3) of the definition. Then, for each $\beta\in B$,  
\begin{equation} \label{phi b}
\phi_\beta:\begin{cases}
x \mapsto x \\
y \mapsto y+\beta
\end{cases}
\end{equation}
is an element of $ \mathrm{Aut}_{D,\lambda P^*}(\mathcal{X}^*)$, if $\phi_\beta(D)=D$.
\end{lemma}
\begin{proof}
Note that, $
\phi_\beta\in \mathrm{Aut}(\mathcal{X}^*)$ and $\phi_\beta(P^*)=P^*$,  for all $\beta\in B$. Then the result follows immediately from the definition of  $ \mathrm{Aut}_{D,\lambda P^*}(\mathcal{X}^*)$
\end{proof}

Although the problem of calculating the group of automorphisms for SAP curves in general is a difficult problem, we identified that some SAP curves have sets of automorphisms of interest for obtaining r-PD sets for some burst errors. In this sense we will call \textit{special SAP curve}, and denote them by $\mathcal{X}^{**}$, those that meet the following conditions:

\begin{enumerate}
\item[(1*)] $g(0)=0$ and $g(-a)=g(a)$ for all $a\in \mathbb{F}_{q}^{*}$;
\item[(2*)] if $(a,b)\in\mathcal{X}^{**}(\mathbb{F}_q)$, then
    \begin{enumerate}
        \item there exists $\omega_a (x) \in \mathbb{F}_q[x]$ such that $\omega_a(0)=0$, $\omega_a (-a)=-\omega_a(a)$  and $f(\omega_a(a))=2g(a)$.
        \item there exists a divisor $D$ such that $P^* \notin \mathrm{Supp} (D)$ and $\phi_{a,b}^{\omega_a } \in \mathrm{Aut}_{D,\lambda P^*}(\mathcal{X}^{**})$ induced by $(a,b)$ and $\omega_a $ given by
        $$
\phi_{a,b}^{\omega_a }:\begin{cases}
x \mapsto x+a, \\
y \mapsto y +  \omega_a (x)+b.
\end{cases}
$$
    \end{enumerate}  
\end{enumerate}

Note that, for a special SAP curve we have that $(0,0)\in\mathcal{X}^{**}(\mathbb{F}_q)$, and that $(a,b)\in\mathcal{X}^{**}(\mathbb{F}_q)$ iff $ (-a,b)\in\mathcal{X}^{**}(\mathbb{F}_q).$ 
In Section \ref{PD}, we will see that special SAP curves allow us to establish particularly new sets for permutation decoding of one-point AG codes arising from these curves (see Theorems \ref{pagen} and \ref{thm3gen}).

Some examples of special SAP curves include: 
\begin{itemize}[leftmargin=*, labelsep=0.5em]
\item The Hermitian curves $\mathcal{H}_q$,  take $\omega_a (x)=a^qx$, see \cite{S1973}.
\item The generalized Hermitian curves $\mathcal{X}_{q,r}$, take $\omega_{a}(x)=\sum_{i=0}^{r-1} (-1)^i (a^{q^r}x)^{q^{r-i-1}}$.
A routine but lengthy computation yields
$\omega_{a}(0)=0$,  $\omega_{a}(-a)=-\omega_{a}(a)$, $ f(\omega_{a}(a))=2g(a)$     
and, by \cite[Section II]{KKO2001}, we know that for all $(a,b)\in \mathcal{X}_{q,r}(\mathbb{F}_{q^{2r}})$,
 $$
\phi_{a , b}^{\omega_a}:\begin{cases}
x \mapsto x+a, \\
y \mapsto y + \omega_{a}(x)+b
\end{cases}
$$
form the subgroup of $\mathrm{Aut}(\mathcal{X}_{q,r})$ fixing the point at infinity $P^*$, see \cite[Section~II]{KKO2001}.
\item The Abdón-Torres curves $\mathcal{Y}_{2}$, take $\omega_a (x)=a^{2q}x^2+a^qx$. 
A straightforward computation shows that
$\omega_a(0)=0$, $\omega_a(-a)=\omega_a(a)=-\omega_a(a)$,  $f(\omega_a(a))=0=2g(a)$
and by \cite[Proposition 3.1]{HM2026}, for all $(a,b)\in \mathcal{Y}_2(\mathbb{F}_{q^{2}})$, 
 $$
\phi_{a,b}^{\omega_a}:\begin{cases}
x \mapsto x+a, \\
y \mapsto y+a^{2q}x^2+a^qx+b
\end{cases}
$$
is an automorphism of $C(D,\lambda P_\infty)$, where $D$ is the sum of all $\frac{q^3}{2}$ rational affine points of $\mathcal{Y}_2$ and $P_\infty$ is the only rational point at infinity of the curve $\mathcal{Y}_2$. 
\end{itemize}

\section{Information positions for $C(D,\lambda P^*)$}
In this section, we consider one-point AG codes $C(D,\lambda P^*)$ arising from a SAP curve $\mathcal{X}^*$. We use the orbits of $\sigma$ to obtain a description of information positions for $C(D,\lambda P^*)$, following the approach introduced in \cite{FT2018} via the root diagram of these codes.

Consider the automorphism $\sigma \in \mathrm{Aut}_{D,\lambda P^*}(\mathcal{X^*})$ given by
\[
\sigma :
\begin{cases}
x \mapsto \alpha x, \\
y \mapsto  \alpha^t y,
\end{cases}
\]
for some positive integer $t$ and some $\alpha \in \mathbb{F}_q^*$. Assume that $\mathrm{ord}(\alpha) = \nu$ and $\mathrm{ord}(\alpha^t) = \mu$, with $t\mu=\nu$.

Suppose that the support of $D$ decomposes into disjoint orbits under the action of $\sigma$ as
\[
\mathrm{Supp}(D) = O_1 \cup \cdots \cup O_r \cup O_{r+1} \cup \cdots \cup O_{r+s}.
\]
For each $i=1,\ldots,r+s$, we denote $O_i = \{ P_{i,0}, P_{i,1}, \ldots, P_{i,\# O_i-1}  \}$, where $P_{i,0}=(x_i,y_i) \in \mathcal{X}^*(\mathbb{F}_q)$ and $P_{i,j}=\sigma^j(P_{i,0})=(\alpha^j x_i , \alpha^{jt}y_i)$, for $j=1, \ldots , \# O_i-1$.

Also, suppose that the last $s$ orbits $O_{r+1},\ldots,O_{r+s}$ are exactly those that contain $\mathbb{F}_q$-rational points on $\mathcal{X}^*$ having at least one zero coordinate, that is, $Q \in \mathcal{X}^*(\mathbb{F}_q)$ such that $Q=(0,\eta)$ or $Q=(\omega,0)$. 

Note that, by definition of $\sigma$, we have the following.
\begin{itemize}[leftmargin=*, labelsep=0.5em]
    \item If $(0,0)\in O_i$, then 
    $O_i=\{(0,0)\}$ is an orbit of size 1.
    \item If $(0,\eta) \in O_i$, for some $\eta \in \mathbb{F}_q^*$, then
$$
O_i = \{(0,\eta), (0,\alpha^t\eta), \ldots, (0,\alpha^{t(\mu-1)}\eta)\},
$$  
is an orbit with $\# O_i = \mu$.
    \item If $(\omega,0) \in O_i$, for some $\omega \in \mathbb{F}_q^*$, then
$$
O_i = \{(\omega,0), (\alpha\omega,0), \ldots, (\alpha^{\nu-1}\omega,0)\}.
$$
is an orbit with $\# O_i = \nu$.
\end{itemize}

We consider the first $r$ orbits $O_{1},\ldots,O_{r}$, namely those contain $\mathbb{F}_q$-rational points on $\mathcal{X}^*$ with nonzero coordinates.  Let $B=\{\beta \in \mathbb{F}_q : f(\beta)=0\} =\{0=\beta_1,\beta_2, \ldots, \beta_{\ell}\}$ be as in Lemma \ref{tal}, where $\ell = \# B$. Note that,
if $ (x_i,y_i)\in O_i$, with $x_i,y_i\in\mathbb{F}_q^*$, 
then 
$O_i=\{(\alpha^{k-1}x_i,\alpha^{t(k-1)}y_i): k\in[\nu]\}.$
Consequently, due to the structure of $O_i$, there is exactly one $\mathbb{F}_q$-rational point in $O_i$ having a given first component $x_i$, namely,
$O_i\cap\mathcal{P}_{x_i}=\{(x_i,y_i)\}$. Now, by Lemma \ref{tal}(1), if $(a,b)\in \mathcal{X}^*(\mathbb{F}_q)$, then, for all $\beta \in B$,  $(a,b+\beta)\in \mathcal{X}^*(\mathbb{F}_q)$. So, in addition to the orbit $O_i$, there are $\ell - 1$ other orbits that have one $\mathbb{F}_q$-rational point of the form $(x_i,  y_i + \beta)$ for some $\beta \in B$. Then, we can conclude that $r = m \ell$ for some positive integer $m$.

\begin{assumption}\label{ass:orbits}
We assume henceforth that $\sigma$ acts on $\mathcal{X}^*(\mathbb{F}_q)$ in such a way that the $r$ orbits with nonzero coordinates can be labelled so that, for $i=1,\ldots,m$ and $j=1,\ldots,\ell$:
\begin{equation}\label{orbgen} 
O_{i + (j - 1) m} = \{ (\alpha^{k-1}x_{i}, \alpha^{t(k-1)}(y_i + \beta_j ))\mbox{ ; } k \in [\nu] \}.
\end{equation}
\end{assumption}

\begin{remark}
Assumption~\ref{ass:orbits} means that $\sigma$ maps each $y_i$ to $\alpha^t y_i$ within the orbit, so the orbits with the same $x$-skeleton but different $\beta$-shifts are related by the maps $\phi_\beta$ of Lemma~\ref{tal}. This holds in all the examples treated below ($\mathcal{H}_{5,3}$, $\mathcal{Y}_2$ with $q=8$ and $\mathcal{X}_{2,3}$). It does \emph{not} hold for every SAP curve: for example, $y^7 - y = x^{13} - x$ over $\mathbb{F}_7$ with $\sigma(x,y)=(\alpha x,\alpha y)$ does not satisfy this arrangement in general. The assumption should be verified for each specific curve.
\end{remark}

Note that, under Assumption~\ref{ass:orbits}, the last $m$ orbits are those with $j=\ell$, i.e.\
for $r-m < i \leq r$, $O_i = \{ (\alpha^{k-1}x_{i}, \alpha^{t(k-1)}(y_i + \beta_\ell))\;;\; k \in [\nu] \}$.

\begin{assumption}\label{ass:nu}
We assume that $\nu=\mu\tau_2$; equivalently, for every $(x_i,y_i)\in O_i$ with $1\le i\le r$ the polynomial $g(x)-f(y_i)\in\mathbb{F}_q[x]$ has $\tau_2$ distinct roots in $\mathbb{F}_q$.
\end{assumption}

\begin{remark}
In general one only has $\nu\le\mu\tau_2$, because the $\nu/\mu$ points of $\mathcal{Q}_{y_i}\cap O_i$ (Lemma~\ref{descoi}) have first coordinates among the roots of $g(x)-f(y_i)$, a polynomial of degree $\tau_2$. Assumption~\ref{ass:nu} holds in the examples below: $12=4\cdot 3$ for $\mathcal{H}_{5,3}$, $9=1\cdot 9$ for $\mathcal{X}_{2,3}$ and for $\mathcal{Y}_2$ with $q=8$.
\end{remark}

Moreover, we have the following results.

\begin{lemma}\label{descoi}
    If $(x_i,y_i)\in O_i$ with $1\leq i\leq r$, then
    $$O_i=(\mathcal{Q}_{y_i}\cap O_i)\;\sqcup\;(\mathcal{Q}_{\alpha^t {y_i}}\cap O_i)\;\sqcup\; \cdots \;\sqcup\;(\mathcal{Q}_{\alpha^{t(\mu-1)}{y_i}}\cap O_i)$$
    is a partition of $O_i$ into $\mu$ disjoint sets each of size $\dfrac{\nu}{\mu}$.
\end{lemma}

\begin{proof}
Let $(x_i,{y_i})\in O_i$. So, we have $O_i=\{(\alpha^j x_i,\alpha^{tj}y_i):j=0,1,\dots,\nu-1\}$. A point $(\alpha^j x_i,\alpha^{tj}{y_i})\in O_i$ belongs to $\mathcal{Q}_{\alpha^{tk}y_i}$ if and only if $\alpha^{tj}y_i=\alpha^{tk}y_i$, i.e.\ $\alpha^{t(j-k)}=1$, i.e.\ $\mu\mid (j-k)$ (since $\mathrm{ord}(\alpha^t)=\mu$). Hence
$$\mathcal{Q}_{\alpha^{tk}y_i}\cap O_i = \{(\alpha^{k+\mu n}x_i, \alpha^{t(k+\mu n)}y_i):n=0,1,\dots,\tfrac{\nu}{\mu}-1\},$$
which has exactly $\nu/\mu$ elements and is determined by $k\in\{0,1,\dots,\mu-1\}$. The $\mu$ sets for $k=0,\dots,\mu-1$ are pairwise disjoint and their union is $O_i$, since every $j\in\{0,\dots,\nu-1\}$ has a unique residue $k=j\bmod\mu$.
\end{proof}

\begin{lemma}\label{lem:fibre}
Assume Assumption~\ref{ass:nu}. If $(a,b)\in O_i$ with $1\le i\le r$, then $\mathcal{Q}_b\cap\mathrm{Supp}(D)\subseteq O_i$, that is, $\mathcal{Q}_b\cap\mathrm{Supp}(D)=\{(\alpha^{\mu n}a,b):0\le n<\nu/\mu\}$, a set with exactly $\nu/\mu=\tau_2$ points.
\end{lemma}
\begin{proof}
A point $(x,b)\in\mathcal{X}^*(\mathbb{F}_q)$ satisfies $g(x)=f(b)$, so its first coordinate is a root of $g(x)-f(b)$, a polynomial of degree $\tau_2$; hence $\#\mathcal{Q}_b\le\tau_2$. By Lemma~\ref{descoi}, $\mathcal{Q}_b\cap O_i$ already has $\nu/\mu=\tau_2$ points (Assumption~\ref{ass:nu}). Therefore $\mathcal{Q}_b\cap\mathrm{Supp}(D)=\mathcal{Q}_b\cap O_i$.
\end{proof}

\begin{lemma}\label{exigen}
Let  $B=\{\beta\in\mathbb{F}_q: f(\beta)=0\}$ and  $\ell = \#B$. If $(a,b)\in O_i$ for $1\leq i\leq r$, then there exists $\beta\in B$ such that $(a,b+\beta)\in \bigcup_{k=1}^m O_{k+(\ell-1)m}$.
\end{lemma}
\begin{proof}
  Suppose that $(a,b)\in O_{i+(j-1)m}$, for some $i\in [m]$ and $j\in[\ell]$. Since every rational point with first coordinate $a$ lies in exactly one of the orbits $O_i, O_{i+m},\dots, O_{i+(\ell-1)m}$, by Lemma \ref{tal}, we can conclude that there exists $\beta\in B$ such that $(a,b+\beta)\in O_{i+(\ell-1)m}$. Thus, $(a,b+\beta)\in \bigcup_{i=1}^m O_{i+(\ell-1)m}$.
\end{proof}

Below, we present some examples in which we apply the previous results, and which illustrate the decompositions of $\mathbb{F}_q$-rational points into orbits.

\begin{example}\label{sx53}
Consider the subcover of the Hermitian curve $\mathcal{H}_{5,3}$ with affine equation
$
\quad y^5 + y = x^{3}
$
over $\mathbb{F}_{5^2} = \mathbb{F}_5(\theta)$, where
$\theta$ is a primitive element of $\mathbb{F}_{5^2}$ with $\theta^2 + 4\theta + 2 = 0$ over $\mathbb{F}_5$.
We have
$B = \{0,\theta^{3}, \theta^{9},\theta^{15},\theta^{21}\}$ and
$
\sigma (x,y)=(\theta^{2} x, \theta^{6} y)
$.
So $\ell=5$, $\nu=12$, $\mu=4$ and
the orbits of $\mathcal{H}_{5,3}$ are:
\begin{align*}
O_1& = \{(\theta^{0+(k-1)2}, \theta^{1+(k-1)6}) : k \in [12]\}
= \mathcal{Q}_{\theta} \cup \mathcal{Q}_{\theta^{7}} \cup \mathcal{Q}_{\theta^{13}}\cup \mathcal{Q}_{\theta^{19}}\\
O_2& = \{(\theta^{0+(k-1)2}, \theta^{18+(k-1)6}) : k \in [12]\}
= \mathcal{Q}_{\theta^{18}} \cup \mathcal{Q}_{\theta^{24}} \cup \mathcal{Q}_{\theta^{6}}\cup \mathcal{Q}_{\theta^{12}}\\
O_3& = \{(\theta^{0+(k-1)2}, \theta^{5+(k-1)6}) : k \in [12]\}
= \mathcal{Q}_{\theta^{5}} \cup \mathcal{Q}_{\theta^{11}} \cup \mathcal{Q}_{\theta^{17}}\cup \mathcal{Q}_{\theta^{23}}\\
O_4& = \{(\theta^{0+(k-1)2}, \theta^{4+(k-1)6}) : k \in [12]\}
= \mathcal{Q}_{\theta^{4}} \cup \mathcal{Q}_{\theta^{10}} \cup \mathcal{Q}_{\theta^{16}}\cup \mathcal{Q}_{\theta^{22}}\\
O_5& = \{(\theta^{0+(k-1)2}, \theta^{20+(k-1)6}) : k \in [12]\}
= \mathcal{Q}_{\theta^{20}} \cup \mathcal{Q}_{\theta^{2}} \cup \mathcal{Q}_{\theta^{8}}\cup \mathcal{Q}_{\theta^{14}}\\
O_6&= \{(0,\beta) : \beta\in B\setminus\{0\}\}
= \mathcal{Q}_{\theta^{3}} \cup \mathcal{Q}_{\theta^{9}} \cup \mathcal{Q}_{\theta^{15}}\cup \mathcal{Q}_{\theta^{21}}\\
O_7&= \{(0,0)\}
= \mathcal{Q}_{0}.
\end{align*}
\end{example}

\begin{example}\label{y2q8}
Consider the Abdón-Torres curve with $q=8$ defined by the affine equation
$$y^4 + y^2 + y = x^{9},$$
over $\mathbb{F}_{64}=\mathbb{F}_2(\theta)$, where $\theta^6+\theta^4+\theta^3+\theta+1=0$. We have $\tau_1=4$, $\tau_2=9$, $g=12$, $\nu=9$, $\mu=1$ (so $\nu=\mu\tau_2$), $B=\{0,\theta^9,\theta^{18},\theta^{36}\}$ and $\sigma(x,y)=(\theta^7x,y)$. Thus $\ell=4$, $r=28$, $m=7$ and, since $\mu=1$, every point $(0,\eta)$ is fixed by $\sigma$, so $s=4$. Writing $\beta_1=0$, $\beta_2=\theta^9$, $\beta_3=\theta^{18}$, $\beta_4=\theta^{36}$, the orbits with nonzero coordinates are, for $i\in[7]$ and $j\in[4]$,
$$O_{i+7(j-1)}=\{(\theta^{i+7(k-1)},\theta^{e_{i,j}}):k\in[9]\}=\mathcal{Q}_{\theta^{e_{i,j}}},\qquad \theta^{e_{i,j}}=y_i+\beta_j,$$
where the exponents $e_{i,j}$ (with $\theta^0=1$) are
\begin{center}
\begin{tabular}{c|ccccccc}
 & $i=1$ & $i=2$ & $i=3$ & $i=4$ & $i=5$ & $i=6$ & $i=7$\\ \hline
$j=1$ & 15 & 7 & 4 & 14 & 2 & 1 & 0\\
$j=2$ & 35 & 56 & 34 & 39 & 10 & 8 & 27\\
$j=3$ & 28 & 30 & 20 & 49 & 16 & 40 & 54\\
$j=4$ & 57 & 51 & 32 & 60 & 17 & 5 & 45
\end{tabular}
\end{center}
that is, $O_1=\mathcal{Q}_{\theta^{15}}$, $O_2=\mathcal{Q}_{\theta^{7}}$, \dots, $O_{28}=\mathcal{Q}_{\theta^{45}}$ in the order of the table read by rows. The last orbits are the singletons
$$O_{29}=\{(0,\theta^9)\},\quad O_{30}=\{(0,\theta^{18})\},\quad O_{31}=\{(0,\theta^{36})\},\quad O_{32}=\{(0,0)\}.$$
Indeed, $28\cdot 9+4=256$ is the number of affine rational points of $\mathcal{Y}_2$.
\end{example}

\begin{example}\label{x23}
Consider the curve
$
\mathcal{X}_{2,3} : y^2 + y = x^9
$
over $\mathbb{F}_{64}$. Suppose that $\mathbb{F}_{64}^*=\langle \theta \rangle$, where $\theta^6 + \theta^4 + \theta^3 + \theta + 1 = 0$ over $\mathbb{F}_2$.
We have
$
\nu = 9, \mu = 1
$
and 
$
B = \{0,1\}.
$
Then the automorphism is
$
\sigma (x,y)=(\theta^7 x, y)
$.
The orbits are:
\begin{align*}
O_1& = \{(\theta^{1+(k-1)7}, \theta^{18}) : k \in [9]\}
=\mathcal{Q}_{\theta^{18}}
\\
O_2& = \{(\theta^{2+(k-1)7}, \theta^{36}) : k \in [9]\}
=\mathcal{Q}_{\theta^{36}}
\\
O_3& = \{(\theta^{3+(k-1)7}, \theta^{31}) : k \in [9]\}
=\mathcal{Q}_{\theta^{31}}
\\
O_4&= \{(\theta^{4+(k-1)7}, \theta^{9}) : k \in [9]\}
=\mathcal{Q}_{\theta^{9}}
\\
O_5&= \{(\theta^{5+(k-1)7}, \theta^{47}) : k \in [9]\}
=\mathcal{Q}_{\theta^{47}}
\\
O_{6}& = \{(\theta^{6+(k-1)7}, \theta^{55}) : k \in [9]\}
=\mathcal{Q}_{\theta^{55}}
\\
O_{7}& = \{(\theta^{7+(k-1)7}, \theta^{21}) : k \in [9]\}
=\mathcal{Q}_{\theta^{21}}
\\
O_8&=\{(\theta^{1+(k-1)7}, \theta^{54}) : k \in [9]\}
=\mathcal{Q}_{\theta^{54}}
\\
O_9& = \{(\theta^{2+(k-1)7}, \theta^{45}) : k \in [9]\}
=\mathcal{Q}_{\theta^{45}}
\\
O_{10}&= \{(\theta^{3+(k-1)7}, \theta^{59}) : k \in [9]\}
=\mathcal{Q}_{\theta^{59}}
\\
O_{11}& = \{(\theta^{4+(k-1)7}, \theta^{27}) : k \in [9]\}
=\mathcal{Q}_{\theta^{27}}
\\
O_{12}& = \{(\theta^{5+(k-1)7}, \theta^{61}) : k \in [9]\}
=\mathcal{Q}_{\theta^{61}}
\\
O_{13}& = \{(\theta^{6+(k-1)7}, \theta^{62}) : k \in [9]\}
=\mathcal{Q}_{\theta^{62}}
\\
O_{14}& = \{(\theta^{7+(k-1)7}, \theta^{42}) : k \in [9]\}
=\mathcal{Q}_{\theta^{42}}
\\
O_{15} &= \{(0,1)\}
\\
O_{16} &= \{(0,0)\}
\end{align*}
\end{example}

The decomposition of the support of $D$ into disjoint orbits under the action of $\sigma$ as above yields the following result which is a consequence of \cite[Proposition 3.1]{FT2018} and Lemmas \ref{descoi} and \ref{lem:fibre}.

\begin{lemma}\label{lemma:bi_mi}
    Let $i\in [r]$ and $(a,b)\in O_i$. Then
    \begin{enumerate}
               \item When $\mu=1$, we have $B_i(x,y)=\prod_{j=2}^{\nu} (x-\alpha^{(j-1)}a)$, which vanishes at each point of $O_i$ except $(a,b)$. When $\mu>1$, the function $B_i(x,y)$ defined in \cite[Proposition~3.1]{FT2018} vanishes at each point of $O_i$ except $(a,b)$ (see \cite{FT2018} for the precise formula, which depends on the orbit structure);
                \item $M_i(y)=\prod_{\ell=1}^\mu (y-\alpha^{t(\ell-1)}b)$ is a function such that the orbit $O_i$ is the intersection of $\mathrm{Supp}(D)$ with the curve $M_i(y)=0$. In addition, $M_i(y)$ is a non-zero constant when restricted to each of the orbits $O_k$ with $k\neq i$.
    \end{enumerate}
\end{lemma}

Thus, since $M_i(y)$ and $B_i(x,y)$ are polynomials and $\operatorname{div}_\infty(x) = \tau_1P^*$ and $\operatorname{div}_\infty(y) = \tau_2P^*$, it follows that 
$$
\mathrm{div}_\infty\left(M_i(y)\right)=(\mu \tau_2)P^*   \text{ and }  \mathrm{div}_\infty\left(B_i(x,y)\right)=\left(\left(\tfrac{\nu}{\mu}-1\right)\tau_1 + (\mu - 1)\tau_2\right)P^*.
$$
Moreover, $M_i(y)=y^{\mu}-b^{\mu}$ has degree $\mu$ in $y$ and $\operatorname{div}_\infty(y)=\tau_2P^*$, so $\operatorname{div}_\infty(M_i(y))=\mu\tau_2P^*=\nu P^*$ by Assumption~\ref{ass:nu}. Moreover, by Lemma~\ref{lem:fibre}, $M_i$ vanishes exactly on $O_i$ among the points of $\mathrm{Supp}(D)$, consistently with part~(2) of the lemma. According to the notation of \cite[Proposition~3.3 and Theorem~3.4]{FT2018}, the nonnegative integers $\rho_1$, $\rho_2$, $\rho_3$ are determined by $\mathrm{div}_\infty(M_i)=(\rho_1\tau_2)P^*$ and $\mathrm{div}_\infty(B_i)=(\rho_2\tau_1+\rho_3\tau_2)P^*$. From the pole orders computed above we obtain $\rho_1=\mu$, $\rho_2=\frac{\nu}{\mu}-1=\tau_2-1$ and $\rho_3=\mu-1$. Consequently, applying the construction from Section~\ref{AG codes section} together with the results in \cite{FT2018}, we obtain the following description for the generator matrix of $C(D,\lambda P^*)$.

\begin{proposition}{\cite[Proposition 3.3 and Theorem 3.4]{FT2018}}\label{infgen}
   Assume that  $C(D,\lambda P^*)$ is a one-point code arising from a SAP curve $\mathcal{X}^*$ with genus $g>0$. Let $\sigma \in \mathrm{Aut}_{D,\lambda P^*}(\mathcal{X^*})$ and $\mathrm{Supp}(D)=O_1\cup \cdots \cup O_{r+s}$ be as above. Recall that $t=\nu/\mu$ and $\tau_2=\nu/\mu=t$. For $1\leq i\leq r$, we have the following.
\begin{enumerate}
   \item If $\nu \geq 2g$ and $\lambda=i\nu+2g-1+l$ with $0\leq l\leq \nu - 2g$, then 
\begin{enumerate}
    \item all points in $O_1, \dots, O_i$ and the first $g+l$ points in $O_{i+1}$ are information positions of $C(D,\lambda P^*)$;
    \item all points in $O_{i+2}, \dots, O_{r+s}$ and the last $\nu-(g+l)$ points in $O_{i+1}$  are check positions of $C(D,\lambda P^*)$.
\end{enumerate}

\item If $\lambda=(i-1)\nu+l_1\tau_1+l_2\tau_2$ with $0\leq l_1\leq t-1$ and $0\leq l_2\leq \mu-1$, and
$$E_i^\lambda:=\{(l_1,l_2)\in\mathbb{N}_0^2 : 0\leq l_1\leq t-1, \ 0\leq l_2\leq \mu-1, \ (i-1)\nu+l_1\tau_1+l_2\tau_2\leq \lambda \},$$
$\kappa_i=\#E_i^\lambda$ and $\tau=\left\lceil \frac{(t-1)\tau_1+(\mu-1)\tau_2}{\nu} \right\rceil$, then
\begin{enumerate}
    \item all points in $O_1, \dots, O_{i-\tau}$ and the first $\kappa_{i-\tau+1},\dots, \kappa_{i}, \kappa_{i+1}, \dots \kappa_{i+\tau-1}$ points in $O_{i-\tau+1},\dots, O_{i}, O_{i+1}, \dots O_{i+\tau-1}$ respectively, are information positions of $C(D,\lambda P^*)$;
    \item all points in $O_{i+\tau}, \dots, O_{r+s}$ and the last $\nu-\kappa_{i-\tau+1},\dots, \nu-\kappa_{i}, \nu-\kappa_{i+1}, \dots \nu-\kappa_{i+\tau-1}$ points in $O_{i-\tau+1},\dots, O_{i}, O_{i+1}, \dots O_{i+\tau-1}$  respectively, are check positions of $C(D,\lambda P^*)$.
\end{enumerate}
\end{enumerate}
\end{proposition}

Equivalently, in part~(2), for every $j\in[r]$ the information positions in $O_j$ are its first $\kappa_j$ points, where $\kappa_j=\#E_j^\lambda$ with $E_j^\lambda$ defined by the same formula (so $\kappa_j=\nu$ for $j\le i-\tau$ and $\kappa_j=0$ for $j\ge i+\tau$); this is how the proposition is used in Example~\ref{y2q8d}.

This characterization of information and check positions will be used in the following section to get a permutation decoding.
Furthermore, we use $\mathcal{I}$ to refer to the set of points on a curve that are associated with information positions of a code, and we use $\mathcal{H}$ to refer to the set of points on a curve that are associated with check positions of a code.

\begin{example}
    Consider the one-point $C(D,30P_\infty)$ over the subcover of the Hermitian curve $\mathcal{H}_{5,3}$ given in the Example \ref{sx53}. Note that, $\tau_1=5$, $\tau_2=3$, $\mu=4$, $\nu=12$ and $t=\frac{\nu}{\mu}=3$. So, since $\lambda=30=2\cdot 12+0\cdot 5+2\cdot 3$ and $\tau=\left\lceil\frac{(t-1)\tau_1+(\mu-1)\tau_2}{\nu}\right\rceil=\left\lceil\frac{2\cdot5+3\cdot3}{12}\right\rceil=\left\lceil\frac{19}{12}\right\rceil=2$, we get $E_1^{30}=\{0,1,2\}\times\{0,1,2,3\}$, $E_2^{30}=\{0,1,2\}\times\{0,1,2,3\}\setminus \{(2,3)\}$, $E_3^{30}=\{(0,0),(0,1),(0,2),(1,0)\}$ and $E_4^{30}=E_5^{30}=\emptyset$. Therefore,
    \begin{align*}
        \mathcal{I}&=O_1\cup \left(O_2\setminus \{(\theta^{22},\theta^{12})\}\right)\cup \{(1,\theta^{5}),(\theta^{2},\theta^{11}),(\theta^{4},\theta^{17}),(\theta^{6},\theta^{23})\} \text{ and }\\
         \mathcal{H}&=\mathrm{Supp}(D)\setminus \mathcal{I}. 
    \end{align*}

A direct computation shows that the evaluation matrix of a basis of $\mathcal{L}(30P_\infty)$ restricted to the points $\mathcal{I}$ has full rank equal to 27.
\end{example}

\section{Permutation decoding}\label{PD}
In this section, we continue using the notation introduced in the previous one. We demons\-trate how permutation decoding can correct burst errors in one-point AG codes arising from SAP curves $\mathcal{X}^*$. The following theorem shows how permutation decoding corrects errors in the positions associated with points that have the same second component.

\begin{lemma}\label{lem:rt}
Let $C(D,\lambda P^*)$ be a one-point code on a SAP curve $\mathcal{X}^*$ with
$\lambda=i\nu+2g-1+l$, where $0\le l\le \nu-2g$ and $i\le r-m-1$
(so $\nu\ge 2g$). Let $n=r\nu+s'$, where $s'=\sum_{k=1}^{s}\#O_{r+k}$,
and let $d$ be the minimum distance of $C(D,\lambda P^*)$. Then
\[
d\ \ge\ n-\lambda\ \ge\ m\nu+s'+1 .
\]
Consequently, if
\begin{equation}\label{eq:cond-rt}
m\nu+s'\ \ge\ \frac{2\nu}{\mu}+2,
\end{equation}
then $\dfrac{\nu}{\mu}<\Big\lfloor\dfrac{d-1}{2}\Big\rfloor$.
\end{lemma}

\begin{proof}
Since $i\le r-m-1$ and $l\le \nu-2g$,
\[
\lambda\le (r-m-1)\nu+2g-1+(\nu-2g)=(r-m)\nu-1 .
\]
Hence $d\ge n-\lambda\ge r\nu+s'-(r-m)\nu+1=m\nu+s'+1$.
If \eqref{eq:cond-rt} holds, then $d\ge 2\nu/\mu+3$, so
$\lfloor (d-1)/2\rfloor\ge \nu/\mu+1>\nu/\mu$.
\end{proof}

\begin{theorem}\label{qbgen}
Let $C(D,\lambda P^*)$ be a one-point code on the SAP curve $\mathcal{X}^*$,
and let $B=\{\beta\in\mathbb{F}_q: f(\beta)=0\}$, $\ell=\#B$. Assume that
\begin{enumerate}[label=\textup{(\roman*)}]
\item $\lambda$ is as in Proposition~\ref{infgen}(1) with $i\le r-m-1$
      (so that $O_{i+2},\dots,O_{r+s}$ are check orbits);
\item $\dfrac{\nu}{\mu}<\Big\lfloor\dfrac{d-1}{2}\Big\rfloor$, where $d$ is the minimum
      distance of $C(D,\lambda P^*)$ (by Lemma~\ref{lem:rt}, this holds whenever
      $m\nu+s'\ge 2\nu/\mu+2$);
\item $\phi_\beta(D)=D$ for every $\beta\in B$, where $\phi_\beta$ is as in~\eqref{phi b}.
\end{enumerate}
Then $S=\{\phi_\beta:\beta\in B\}$ is a partial $\frac{\nu}{\mu}$-PD set of size $\ell$
for burst errors impacting positions indexed by $\mathcal{Q}_b$, for every
$b\in\mathbb{F}_q$. That is, for every $b\in\mathbb{F}_q$ and every
$E\subseteq \mathcal{Q}_b\cap \mathrm{Supp}(D)$ (which has at most $\nu/\mu$ points),
there is $\beta\in B$ with $\phi_\beta(E)\subseteq\mathcal{H}$.
Moreover, writing $M=\bigcup_{k=1}^{m}O_{k+(\ell-1)m}$ (a set of $m\nu$ points), for each
$\beta\in B$ the permutation $\phi_\beta$ moves to check positions every error
pattern supported in $\phi_\beta^{-1}(M)$.
\end{theorem}

\begin{proof}
Let $b\in\mathbb{F}_q$ and $(a,b)\in\mathcal{Q}_b\cap\mathrm{Supp}(D)$.

\emph{Case 1: $b\in B$.} Since $B$ is closed under negation, $\phi_{-b}\in S$ and
$\phi_{-b}(a,b)=(a,0)$. By (iii), $(a,0)\in\mathrm{Supp}(D)$, and it has a zero coordinate,
so it lies in $\bigcup_{k=1}^{s}O_{r+k}\subseteq\mathcal{H}$ by (i). The same $\phi_{-b}$
works for all points of $\mathcal{Q}_b$, hence
$\phi_{-b}(\mathcal{Q}_b\cap\mathrm{Supp}(D))\subseteq\mathcal{H}$.

\emph{Case 2: $b\notin B$.} Then $a\ne0$ (otherwise $f(b)=g(0)=0$ would give $b\in B$), and
$b\neq 0$, so $(a,b)\in O_j$ for some $1\le j\le r$. By Lemma~\ref{exigen}
there is $\beta\in B$ with $(a,b+\beta)\in M$. Every point of
$\mathcal{Q}_b\cap\mathrm{Supp}(D)$ lies in $O_j$ (Lemma~\ref{lem:fibre}) and has the form
$(\alpha^{\mu n}a,\,b)=\sigma^{\mu n}(a,b)$. Since $\alpha^{t\mu}=1$, we get
$\phi_\beta(\alpha^{\mu n}a,b)=(\alpha^{\mu n}a,b+\beta)=\sigma^{\mu n}(a,b+\beta)$,
which lies in the same orbit as $(a,b+\beta)$, hence in $M$.
Thus $\phi_\beta(\mathcal{Q}_b\cap\mathrm{Supp}(D))\subseteq M\subseteq\mathcal{H}$, with the
same $\beta$ for all its points.

Finally, by (i) and (ii), $\nu/\mu<\lfloor (d-1)/2\rfloor$, so the test of Lemma~\ref{lem:PD}
applies to the permuted received word; the last assertion is immediate from
$\phi_\beta(\phi_\beta^{-1}(M))=M\subseteq\mathcal{H}$.
\end{proof}

\begin{remark}\label{rem:partial-check}
We clarify the two key conditions in Theorem~\ref{qbgen}.

\medskip
\noindent (i) Why $i \leq r-m-1$ is needed.
Recall from Proposition~\ref{infgen}(1) that, for $\lambda = i\nu + 2g - 1 + l$ with $0 \leq l \leq \nu - 2g$, the orbit $O_{i+1}$ is \emph{mixed}: its first $g+l$ points are information positions and its last $\nu-(g+l)$ points are check positions. All orbits $O_{i+2}, \dots, O_{r+s}$ are entirely check orbits.

The proof of Theorem~\ref{qbgen} requires that the $m$ orbits $O_{r-m+1}, \dots, O_r$ (the last $m$ orbits among those with nonzero coordinates, in the arrangement~\eqref{orbgen}) are \emph{entirely} check orbits. By Proposition~\ref{infgen}(1), this holds if and only if $i+2 \leq r-m+1$, that is, $i \leq r-m-1$. If instead $i = r-m$, then $O_{i+1} = O_{r-m+1}$ is the mixed orbit, which contains $g+l$ information positions; since $g \geq 1$ and $l \geq 0$, at least one of the $m$ target orbits would contain information positions, and the displacement argument in the proof breaks down.

\medskip
\noindent (ii) Why the set is a \emph{partial} PD set, and why Lemma~\ref{lem:PD} applies.
The set $S = \{\phi_\beta : \beta \in B\}$ is a partial $r'$-PD set with $r' = \nu/\mu$ (the size of each $\mathcal{Q}_b$), restricted to burst errors supported on a single $\mathcal{Q}_b$. To apply the decoding procedure of Lemma~\ref{lem:PD} (i.e., \cite[Theorem~8.1]{H1998}), one must have $r' < t := \lfloor(d-1)/2\rfloor$, where $d$ is the minimum distance of $C(D, \lambda P^*)$.

Since $C(D, \lambda P^*)$ is an AG code, its minimum distance satisfies $d \geq n - \lambda$. Under the hypothesis $\lambda = i\nu + 2g - 1 + l$ with $i \leq r-m-1$, we have $n = r\nu + s'$ (where $s'$ counts the special orbit sizes) and $\lambda \leq (r-m-1)\nu + 2g - 1 + (\nu - 2g) = (r-m)\nu - 1$, so
$$d \geq n - \lambda \geq m\nu + s' + 1 > \nu \geq \frac{\nu}{\mu} = r',$$
which confirms $r' < d-1$, hence $r' \leq t$. (For the strict inequality $r' < t$ needed for a \emph{partial} PD set, one requires $\nu/\mu < \lfloor(d-1)/2\rfloor$, i.e.\ $d \geq 2\nu/\mu + 3$. By Lemma~\ref{lem:rt} this holds as soon as $m\nu+s'\ge 2\nu/\mu+2$; otherwise it has to be verified directly for the parameters at hand).
\end{remark}

\begin{example}\label{ex:qb}
Consider the subcover of the Hermitian curve $\mathcal{H}_{5,3}: y^5+y=x^{3}$ over $\mathbb{F}_{5^2}$ as in Example~\ref{sx53}, with $D$ the sum of its $65$ affine rational points, and the code $C(D,47P_\infty)$. Here $r=5$, $m=1$, $\ell=5$, $g=4$, $\nu=12$, $\mu=4$ and $s'=5$, and $\lambda=47=3\cdot 12+2\cdot 4-1+4$, that is, $i=3=r-m-1$ and $l=4=\nu-2g$. By Proposition~\ref{infgen}(1), $O_1\cup O_2\cup O_3$ and the first $8$ points of $O_4$ are information positions, while the last four points of $O_4$ and the points of $O_5\cup O_6\cup O_7$ are check positions. The code is a $[65,44,d]_{25}$ code with $d\ge n-\lambda=18$, so $\lfloor(d-1)/2\rfloor\ge 8>\nu/\mu=3$, and $\phi_\beta(D)=D$ because $D$ contains all affine rational points. Hence the hypotheses of Theorem~\ref{qbgen} hold, and $M=O_5$.

For $b\in B=\{0,\theta^{3}, \theta^{9},\theta^{15},\theta^{21}\}$, $\phi_{0}(\mathcal{Q}_b)=\mathcal{Q}_b\subseteq O_{6}\cup O_{7}\subseteq \mathcal{H}$. For $b\in\{\theta^{2},\theta^{8},\theta^{14},\theta^{20}\}$ we also take $\phi_0$, since $\mathcal{Q}_b\subseteq O_5\subseteq\mathcal{H}$. For the remaining sixteen values $b\notin B$,
\begin{align*}
    \phi_{\theta^{3}}(\mathcal{Q}_{\theta^{7}}\cup\mathcal{Q}_{\theta^{12}}\cup\mathcal{Q}_{\theta^{5}}\cup\mathcal{Q}_{\theta^{16}})&=
\mathcal{Q}_{\theta^{2}}\cup\mathcal{Q}_{\theta^{14}}\cup\mathcal{Q}_{\theta^{20}}\cup\mathcal{Q}_{\theta^{8}}=O_{5}\subseteq \mathcal{H},\\
\phi_{\theta^{9}}(\mathcal{Q}_{\theta^{13}}\cup\mathcal{Q}_{\theta^{18}}\cup\mathcal{Q}_{\theta^{11}}\cup\mathcal{Q}_{\theta^{22}})&=
\mathcal{Q}_{\theta^{8}}\cup\mathcal{Q}_{\theta^{20}}\cup\mathcal{Q}_{\theta^{2}}\cup\mathcal{Q}_{\theta^{14}}=O_{5}\subseteq \mathcal{H},\\
\phi_{\theta^{15}}(\mathcal{Q}_{\theta^{19}}\cup\mathcal{Q}_{\theta^{24}}\cup\mathcal{Q}_{\theta^{17}}\cup\mathcal{Q}_{\theta^{4}})&=
\mathcal{Q}_{\theta^{14}}\cup\mathcal{Q}_{\theta^{2}}\cup\mathcal{Q}_{\theta^{8}}\cup\mathcal{Q}_{\theta^{20}}=O_{5}\subseteq \mathcal{H},\\
\phi_{\theta^{21}}(\mathcal{Q}_{\theta}\cup\mathcal{Q}_{\theta^{6}}\cup\mathcal{Q}_{\theta^{23}}\cup\mathcal{Q}_{\theta^{10}})&=
\mathcal{Q}_{\theta^{20}}\cup\mathcal{Q}_{\theta^{8}}\cup\mathcal{Q}_{\theta^{14}}\cup\mathcal{Q}_{\theta^{2}}=O_{5}\subseteq \mathcal{H}.
\end{align*}
Therefore $\{\phi_{0}, \phi_{\theta^{3}}, \phi_{\theta^{9}}, \phi_{\theta^{15}}, \phi_{\theta^{21}}\}$ is a partial $3$-PD set ($3=\nu/\mu=\#\mathcal{Q}_b$) for burst errors impacting positions indexed by $\mathcal{Q}_b$, for every $b\in\mathbb{F}_{25}$.
\end{example}

The following results show that, for special SAP curves $\mathcal{X}^{**}$, permutation decoding can be used to correct burst errors in positions associated with points having the same first component.

\begin{theorem}\label{pagen}
   Let $C(D,\lambda P^*)$ be a one-point code on the special SAP curve $\mathcal{X}^{**}$ and let $B=\{\beta\in\mathbb{F}_{q}: f(\beta)=0\}$, $\ell=\#B$. Assume that
\begin{enumerate}[label=\textup{(\roman*)}]
\item the points with a zero coordinate are check positions, that is, $O_{r+1}\cup\dots\cup O_{r+s}\subseteq\mathcal{H}$. This holds if $\lambda$ is as in Proposition~\ref{infgen}(1) with $i \leq r-1$. It also holds for any $\lambda$ with $\lambda\le n-s'-1$ when these points are listed last: a nonzero function in $\mathcal{L}(\lambda P^*)$ has at most $\lambda$ zeros, so the evaluation at any $n-s'\ge\lambda+1$ points is injective, the first $n-s'$ columns of a generator matrix have rank $k$, and an information set can be chosen among them;
\item $\ell<\lfloor (d-1)/2\rfloor$, where $d\ge n-\lambda$ is the minimum distance (for instance, if $n-\lambda\ge 2\ell+3$);
\item every $a\in\mathbb{F}_q$ is the first coordinate of some rational point of $\mathcal{X}^{**}$.
\end{enumerate}
   For each $a\in \mathbb{F}_{q}$, fix $b_a$ such that $(a,b_a)\in \mathcal{X}^{**}(\mathbb{F}_{q})$. Then
   $$\left\{\phi_{-a,\ \omega_{-a}(-a) - b_a}^{\omega_{-a}}  :  a\in \mathbb{F}_{q}\right\}$$
   is a partial $\ell$-PD set of size $q$ for burst errors impacting positions indexed by $\mathcal{P}_a$, for any $a\in \mathbb{F}_{q}$.
\end{theorem}
\begin{proof}
Let $a\in \mathbb{F}_{q}$. We claim that $(-a,\omega_{-a}(-a) - b_a)\in \mathcal{X}^{**}(\mathbb{F}_{q}).$ Indeed, if $a=0$, then $(-a,\omega_{-a}(-a) - b_a)=(0,-b_0)$, since $\omega_0(0)=0$. Thus, if $q$ is even, we have $(0,-b_0)=(0,b_0)\in \mathcal{X}^{**}(\mathbb{F}_{q})$, and if $q$ is odd, then $f(-b_0)=-f(b_0)=-g(0)=0=g(0)$. Now, if $a\neq 0$, then $f(\omega_{-a}(-a) - b_a)= f(\omega_{-a}(-a)) - f(b_a)= 2g(-a) - g(a)$. Since $\mathcal{X}^{**}$ is a special SAP curve, we have $g(a)=g(-a)$ and so $f(\omega_{-a}(-a) - b_a) = g(-a)$, and it follows that $(-a,\omega_{-a}(-a) - b_a)\in \mathcal{X}^{**}(\mathbb{F}_{q})$. 

Therefore, by definition of special SAP curve, we have
$$\phi_{-a,\ \omega_{-a}(-a) - b_a}^{\omega_{-a}}\in \mathrm{Aut}_{D,\lambda P^*}(\mathcal{X}^{**}).$$ 

Note that, 
\begin{align*}
  \phi_{-a,\omega_{-a}(-a) - b_a}^{\omega_{-a}}(a,b_a)&= (a-a,b_a+\omega_{-a}(a)+\omega_{-a}(-a)-b_a)
  =(0,0)\in \bigcup_{k=1}^sO_{r+k} . 
\end{align*}


Moreover,  for all $\beta\in B$, 
$$\phi_{-a,\omega_{-a}(-a) - b_a}^{\omega_{-a}}(a,b_a+\beta)=(0,\beta)\in \bigcup_{k=1}^sO_{r+k}.$$

Consequently, $$\phi_{-a,\omega_{-a}(-a) - b_a}^{\omega_{-a}}(\mathcal{P}_a)\subseteq \bigcup_{k=1}^sO_{r+k}\subseteq \mathcal{H},$$
where $\mathcal{H}$ is the set of points on $\mathcal{X}^{**}$ associated with check positions of $C(D,\lambda P^*)$.   
\end{proof}

\begin{example}\label{y2q8d}
    Consider the Abdón-Torres curve $\mathcal{Y}_{2}: y^4+y^2+y=x^{9}$ over $\mathbb{F}_{8^2}$ as in Example~\ref{y2q8}, and the one-point code $C(D,245 P_\infty)$, where $D$ is the sum of the $256$ affine rational points. Its dimension is $k=245-g+1=234$ and $d\ge n-\lambda=11$, so $\lfloor(d-1)/2\rfloor\ge 5>\ell=4$. We denote the points of $D$ by $P_1,\dots,P_{256}$ following the ordering induced by the orbits $O_1,\dots,O_{32}$, so that $P_{253}=(0,\theta^9)$, $P_{254}=(0,\theta^{18})$, $P_{255}=(0,\theta^{36})$ and $P_{256}=(0,0)$.

    Since $\nu=9<2g=24$, Proposition~\ref{infgen}(1) does not apply and we use part~(2): $\lambda=245=25\cdot 9+5\cdot 4$, that is, $i=26$, $\tau=\lceil 8\cdot 4/9\rceil=4$, and $\kappa_j=\#\{l_1\in[0,8]:9(j-1)+4l_1\le 245\}$ equals $9$ for $j\le 24$, $8,6,3,1$ for $j=25,26,27,28$ and $0$ for $j\ge 29$. Hence the information set is
    $$\begin{array}{lcl}
    \mathcal{I}&=&O_1\cup\dots\cup O_{24}\cup\{\text{first 8 points of }O_{25}\}\cup\\
    & & \{\text{first 6 of }O_{26}\}\cup
   \{\text{first 3 of }O_{27}\}\cup\{\text{first point of }O_{28}\},
    \end{array}$$
    
    with $234$ points. The first $234$ positions are \emph{not} an information set (positions $225,232,233,$ $234$ are not in $\mathcal{I}$), so we use Remark~\ref{rem:infset} with $I=\mathcal{I}$. On the other hand, $245\le n-s'-1=251$, so the points $P_{253},\dots,P_{256}$ with a zero coordinate are check positions, and the hypotheses of Theorem~\ref{pagen} hold.

    Assume that a codeword $c=\mathrm{ev}_D(f)\in C(D,245 P_\infty)$ is transmitted and that $y=c+e\in\mathbb{F}_{64}^{256}$ is received, where $e=\alpha_1 e_1+\alpha_2 e_{64}+\alpha_3 e_{127}+\alpha_4 e_{190}$. The four error positions are the points of $\mathcal{P}_\theta$: $P_1=(\theta,\theta^{15})\in O_1$, $P_{64}\in O_8$, $P_{127}\in O_{15}$ and $P_{190}\in O_{22}$. Let $\psi:=\phi_{\theta,\,\omega_\theta(\theta)+\theta^{15}}^{\omega_\theta}$, where $\omega_\theta(x)=\theta^{16}x^2+\theta^8x$ and $\omega_\theta(\theta)=\theta^{36}$. Then $\psi(P_1)=P_{256}$, $\psi(P_{64})=P_{253}$, $\psi(P_{127})=P_{254}$ and $\psi(P_{190})=P_{255}$, so $\psi(y)$ coincides with $\psi(c)$ except in its last four coordinates, which are $f(P_{64})+\alpha_2$, $f(P_{127})+\alpha_3$, $f(P_{190})+\alpha_4$ and $f(P_1)+\alpha_1$. Hence $\psi(y)_{\mathcal{I}}$ has no errors, $\mathrm{wt}(H\psi(y)^T)=4\le 5$, and $c$ is recovered as $\psi^{-1}\big(\psi(y)_{\mathcal{I}}\cdot G_{\mathcal{I}}\big)$, where $G_{\mathcal{I}}$ is the generator matrix that is systematic on $\mathcal{I}$.
\end{example}

The following result synthesizes Theorems \ref{pagen} and \ref{qbgen} and thus allows the correction of burst errors at any two positions of the received vector.

\begin{theorem}\label{thm3gen}
Let $C(D,\lambda P^*)$ be a one-point code on the special SAP curve $\mathcal{X}^{**}$ with $\lambda$ as in Proposition~\ref{infgen}(1) with $i\leq r-m-1$. Assume moreover that every $a\in\mathbb{F}_q$ is the first coordinate of some rational point of $\mathcal{X}^{**}$, and that $\phi_\beta(D)=D$ for all $\beta\in B$ and $m\nu+s'\ge 6$ (so that $d\ge m\nu+s'+1\ge 7$ and $2<\lfloor(d-1)/2\rfloor$, by Lemma~\ref{lem:rt}).
   For each $a\in \mathbb{F}_{q}$, fix $b_a$ such that $(a,b_a)\in \mathcal{X}^{**}(\mathbb{F}_q)$. Then
   $$\left\{\phi_{-a,\omega_{-a}(-a) - b_a+\beta}^{\omega_{-a}}  :  a\in \mathbb{F}_{q} \text{ and } \beta\in B\right\}$$
   is a partial $2$-PD set of size $\# B \cdot q$ for burst errors impacting any two positions of $C(D,\lambda P^*)$.
\end{theorem}
\begin{proof}
    Assume $(a,b),(c,d)\in \mathcal{X}^{**}(\mathbb{F}_{q})$.\\
    Case 1: if $a=c$, then the result follows from Theorem \ref{pagen}, choosing $\beta=0$.\\
    Case 2: if $b=d$, then the result follows from Theorem \ref{qbgen}, since taking $(a,b_a)=(0,0)$ we have $\phi_{0,\omega_{0}(0) - 0+\beta}^{\omega_0}=\phi_{0,\beta}^{\omega_0}=\phi_\beta$.\\ 
    Case 3: if $a\neq c$ and $b\neq d$. By Lemma \ref{tal}, $(a,b)=(a,b_a+\beta')$ for some $\beta'\in B$. Put $\psi_a:=\phi_{-a,\omega_{-a}(-a) - b_a}^{\omega_{-a}}\in\mathrm{Aut}_{D,\lambda P^*}(\mathcal{X}^{**})$. Then
\begin{align*}
    \psi_a((a,b)) &=\left(0,\ b_a+\beta'+\omega_{-a}(a)+\omega_{-a}(-a)-b_a\right)=(0,\beta'), \\
    \psi_a((c,d)) &=\left(c-a,\ d'\right),\qquad d':=d+\omega_{-a}(c)+\omega_{-a}(-a) - b_a,
\end{align*}
where we used $\omega_{-a}(a)=-\omega_{-a}(-a)$. Since $\psi_a$ preserves $\mathrm{Supp}(D)$ and $c-a\neq 0$, there are two possibilities. If $d'=0$, then $\psi_a((c,d))$ has a zero coordinate, so it lies in $\bigcup_{k=1}^sO_{r+k}$, and we take $\beta=0$. If $d'\neq 0$, then $\psi_a((c,d))$ has nonzero coordinates, so it lies in some $O_j$ with $j\le r$ and, by Lemma \ref{exigen}, there exists $\beta \in B$ such that $\left(c-a,d'+\beta\right)\in \bigcup_{k=1}^m O_{k+(\ell-1)m}$. In both cases
\begin{align*}
   \phi_{\beta}\circ \psi_a((a,b)) =(0,\beta'+\beta)=(0,\beta'')\in \bigcup_{k=1}^sO_{r+k},\\ \phi_{\beta}\circ \psi_a((c,d)) =\left(c-a,d'+\beta\right)\in \bigcup_{k=1}^m O_{k+(\ell-1)m}\cup\bigcup_{k=1}^sO_{r+k},
\end{align*}
where $\beta''=\beta'+\beta \in B$. Moreover, $\phi_{\beta}\circ \psi_a=\phi_{-a,\omega_{-a}(-a) - b_a+\beta}^{\omega_{-a}}$, so this permutation maps the two error positions into
$$\left(\bigcup_{k=1}^m O_{k+(\ell-1)m}\right) \bigcup \left(\bigcup_{k=1}^sO_{r+k}\right)\subseteq \mathcal{H},$$
where $\mathcal{H}$ is the set of points on $\mathcal{X}^{**}$ associated with check positions of $C(D,\lambda P^*)$ (recall $i\le r-m-1$).
\end{proof}

\begin{remark}
The hypothesis $i\le r-m-1$ (together with Proposition~\ref{infgen}(1)) is only used to guarantee that $O_{r-m+1},\dots,O_{r+s}$ are check orbits. For the Abdón-Torres curve $\mathcal{Y}_2$ over $\mathbb{F}_{64}$, where Proposition~\ref{infgen}(1) is not available because $\nu<2g$, a computer check shows that the set of Theorem~\ref{thm3gen} moves every pair of positions to check positions for $\lambda=180,188,189,200$ (and fails for $\lambda=245$).
\end{remark}

\begin{example}\label{x23f}
    Consider the curve $\mathcal{X}_{2,3}: y^2+y=x^{9}$ over $\mathbb{F}_{2^6}$ as in Example \ref{x23}. It is the generalized Hermitian curve with $f(y)=y^2+y$ and $g(x)=x^9$, so it is a special SAP curve $\mathcal{X}^{**}$ with $\omega_a(x)=(a^{8}x)^{4}+(a^{8}x)^{2}+a^{8}x$ (we are in characteristic $2$, so $-a=a$ and $\omega_{-a}=\omega_a$). Here $B=\{0,1\}$, $\ell=2$, the genus is $4$, $\nu=9$, $\mu=1$, $r=14$, $m=r/\ell=7$ and $s=2$. The divisor $D$ is the sum of the $128$ affine rational points of $\mathcal{X}_{2,3}$; hence every $a\in\mathbb{F}_{64}$ is the first coordinate of a rational point (of exactly two, since $\#\mathcal{P}_a=\#B=2$), and $\phi_\beta(D)=D$ for $\beta\in B$.

    Take $\lambda=62=6\cdot 9+2\cdot 4-1+1$, that is, $i=6=r-m-1$ and $l=1$ in Proposition~\ref{infgen}(1) (note that $\nu=9\ge 2g=8$ and $0\le l\le \nu-2g=1$). Then $C(D,62P_\infty)$ is a $[128,59,d]_{64}$ code with $d\ge n-\lambda=66$, so $\lfloor (d-1)/2\rfloor\ge 32$. By Proposition~\ref{infgen}(1), the information positions are the points of $O_1,\dots,O_6$ together with the first five points of $O_7$, i.e.\ $P_1,\dots,P_{59}$; the check positions are the last four points of $O_7$, the $63$ points of $O_8\cup\dots\cup O_{14}=\bigcup_{k=1}^{7}O_{k+(\ell-1)m}$ and the points $P_{127}=(0,1)$ and $P_{128}=(0,0)$ of $O_{15}\cup O_{16}$, in total $69=n-k$ positions.

    All the hypotheses of Theorem~\ref{thm3gen} hold. Fix, for each $a\in\mathbb{F}_{64}$, a point $(a,b_a)\in\mathcal{X}_{2,3}(\mathbb{F}_{64})$. Since $-a=a$, the theorem says that
    $$S=\left\{\phi^{\omega_a}_{a,\ \omega_a(a)+b_a+\beta}\ :\ a\in\mathbb{F}_{64},\ \beta\in\{0,1\}\right\}$$
    is a partial $2$-PD set of size $2\cdot 64=128$ for burst errors at any two positions of $C(D,62P_\infty)$.

    Let us see how it acts on a concrete error pattern. Assume that $c=\mathrm{ev}_D(f)\in C(D,62P_\infty)$ is transmitted and that $y=c+\alpha_1e_1+\alpha_2e_{63}$, with $\alpha_1,\alpha_2\in\mathbb{F}_{64}^*$, is received. The error positions are
    $$P_1=(\theta,\theta^{18})\in O_1\qquad\text{and}\qquad P_{63}=(1,\theta^{21})\in O_7 .$$
    They have different first and different second coordinates, so we are in Case~3 of the proof of Theorem~\ref{thm3gen}, with $(a,b)=(\theta,\theta^{18})$ and $(c,d)=(1,\theta^{21})$.
    \begin{enumerate}
        \item Choose $a=\theta$ and $b_\theta=\theta^{18}$. Then $\omega_\theta(x)=\theta^{32}x^4+\theta^{16}x^2+\theta^{8}x$, $\omega_\theta(\theta)=\theta^{36}+\theta^{18}+\theta^{9}=0$ and $\omega_\theta(1)=\theta^{54}$. For $\beta\in\{0,1\}$ the permutation is
        $$\psi_\beta=\phi^{\omega_\theta}_{\theta,\ \theta^{18}+\beta}:(x,y)\longmapsto\left(x+\theta,\ y+\omega_\theta(x)+\theta^{18}+\beta\right).$$
        \item The first error position goes to $\psi_\beta(P_1)=(0,\beta)$, which is $P_{128}=(0,0)$ if $\beta=0$ and $P_{127}=(0,1)$ if $\beta=1$. Both are in $\bigcup_{k=1}^{s}O_{r+k}\subseteq\mathcal{H}$.
        \item The second error position goes to
        $$\psi_\beta(P_{63})=\left(1+\theta,\ \theta^{21}+\omega_\theta(1)+\theta^{18}+\beta\right)=\left(\theta^{56},\ \theta^{42}+\beta\right).$$
        Exactly one of the two points $(\theta^{56},\theta^{42})$ and $(\theta^{56},\theta^{21})$ lies in $M=O_8\cup\dots\cup O_{14}$, as Lemma~\ref{exigen} predicts: for $\beta=0$ we get $(\theta^{56},\theta^{42})$, the eighth point of $O_{14}=\{(\theta^{7k},\theta^{42}):k\in[9]\}$, which is $P_{125}$, while for $\beta=1$ we get $(\theta^{56},\theta^{21})\in O_7$, which is not in $M$. So we take $\beta=0$ and $\psi:=\psi_0=\phi^{\omega_\theta}_{\theta,\theta^{18}}$.
    \end{enumerate}
    Since $\psi(P_1)=P_{128}$ and $\psi(P_{63})=P_{125}$, the vector $\psi(y)$ coincides with $\psi(c)$ except in the coordinates $125$ and $128$, where $\psi(y)_{125}=f(P_{63})+\alpha_2$ and $\psi(y)_{128}=f(P_1)+\alpha_1$. Therefore $\psi(y)_{[59]}$ has no errors, and $\mathrm{wt}\big(H\psi(y)^T\big)=2\le 32$. Hence $c':=\psi(y)_{[59]}\cdot G\in C(D,62P_\infty)$ and $c=\psi^{-1}(c')$.
\end{example}

\section{Conclusion}
In this paper, we investigated permutation decoding for one-point algebraic geometry codes arising from curves defined by separated polynomials. By exploiting automorphisms of the underlying curves, we obtained permutation automorphisms of the associated AG codes and used them to construct partial permutation decoding sets. Building on the orbit structure induced by suitable automorphisms, we described information and check positions for a family of one-point AG codes and showed how these structures can be incorporated into the permutation decoding framework.

We introduced the class of SAP curves and identified conditions under which their automorphisms yield effective partial PD-sets. For these codes, we proved that permutation decoding can be used to correct specific burst-error patterns associated with rational points sharing the same first or second coordinate. Furthermore, for the subclass of special SAP curves, we established additional decoding sets that allow the correction of burst errors in positions indexed by points with a common first coordinate, culminating in the construction of a partial 2-PD set capable of handling errors occurring at any two rational-point positions.

The results presented here extend previous work on permutation decoding of AG codes from Hermitian and norm-trace curves to a broader family of curves defined by separated polynomials. They demonstrate that the interplay between the automorphism group of a curve and the geometric structure of its rational points provides a powerful tool for the design of efficient decoding strategies.

Future research may focus on extending these techniques to AG codes arising from other families of curves with large automorphism groups, such as other maximal curves, Deligne–Lusztig curves, Suzuki and Ree curves, Giulietti–Korchmáros curves, and related Galois-covered curves. More generally, it would be interesting to identify group-theoretic conditions guaranteeing the existence of efficient PD-sets and to investigate permutation decoding for multipoint AG codes. The construction of smaller PD-sets and a detailed analysis of the computational complexity of the resulting decoding procedures also remain promising directions for further study.

\end{document}